\documentclass[12pt, draftclsnofoot, onecolumn]{IEEEtran}
\usepackage{color} 
\usepackage{fancyhdr}
\usepackage{cite}

\usepackage{graphicx}
\DeclareGraphicsExtensions{.eps}
\usepackage{setspace}
\providecommand{\PSforPDF}[1]{#1}
\usepackage[cmex10]{amsmath}
\usepackage{amsfonts,amssymb,amsthm}
\usepackage{times}
\PSforPDF{\usepackage{psfrag}}
\usepackage{algorithmic}
\usepackage{algorithm}
\usepackage{array}
\usepackage{multirow}
\usepackage{mdwmath}
\usepackage{url}
\usepackage{mdwtab}
\usepackage{eqparbox}
\usepackage{soul}
\usepackage{nicefrac}
\usepackage[normalem]{ulem}
\usepackage{exscale,relsize}

\newtheorem{lemma}{Lemma}

\newtheorem{definition}{Definition}
\newtheorem{proposition}{Proposition}
\begin{document}

\title{Energy Efficiency of Massive MIMO: Coping with Daily Load Variation }
%


\author{
\IEEEauthorblockN { M M Aftab Hossain \IEEEauthorrefmark {1},  Cicek Cavdar \IEEEauthorrefmark {2},  Emil Bj\"ornson \IEEEauthorrefmark {3}, and Riku J\"antti }\IEEEauthorrefmark {1}}
\maketitle

\begin{abstract} 
Massive MIMO is a promising technique to meet the exponential growth of mobile traffic demand. However, contrary to the current systems, energy consumption of next generation networks is required to be load adaptive as the network load varies significantly throughout the day. In this paper, we propose a load adaptive multi-cell massive MIMO system where each base station (BS) adapts the number of antennas to the daily load profile (DLP) in order to maximize the downlink energy efficiency (EE).  In order to incorporate the DLP, the load at each BS is modeled as an $M/G/m/m$ state dependent queue under the assumption that the network is dimensioned to serve a maximum number of users at the peak load. The EE maximization problem is formulated in a game theoretic framework where the number of antennas to be used by a BS is determined through best response iteration. This load adaptive system achieves around 24\% higher EE and saves around 40\% energy compared to a baseline system where the BSs always run with the fixed number of antennas that is most energy efficient at the peak load and that can be switched off when there is no traffic.                                   
\end{abstract}


\section{Introduction}
The goal of next generation cellular networks (NGNs) is to provide thousand fold area capacity 
 \thanks{\\
\IEEEauthorblockA{\IEEEauthorrefmark{1}School of Electrical Engineering, Aalto University, Finland,
Email:\{mm.hossain, riku.jantti\}@aalto.fi} \\
\IEEEauthorblockA{\IEEEauthorrefmark{2}Wireless@KTH, KTH Royal Institute of Technology, Sweden,
 Email: \ { cavdar@kth.se \\}}
\IEEEauthorblockA{\IEEEauthorrefmark{3} Dept. of Electrical Eng. (ISY), Link\"oping University, Sweden,
Email: \ { emil.bjornson@liu.se}} \\}
improvements while keeping the same cost and power consumption as today~\cite{Ericsson}. In order to 
design such highly energy efficient networks, it is important to consider the fact that temporal variation of traffic demand over the day is  significant; recent data in \cite{EarthModel} shows that the daily maximum loads are even $2-10$ times higher than the daily minimum loads. Contemporary cellular networks are designed  focusing on the peak traffic demand and hence the BS with no load  consumes at least half of the energy it requires to serve the peak load~\cite{EarthModel}. The NGNs are required to  adapt their power consumption to the temporal load variation. Massive MIMO systems are expected to play a pivotal role in delivering very high capacity in NGNs. In massive MIMO systems, each BS uses hundreds of antennas to be able to simultaneously serve tens of user equipments (UEs) on the same time-frequency resource, by downlink  multi-user precoding and uplink multi-user detection\cite{Larsson}.  This study aims to give an insight to the energy efficient design of multi-cell massive MIMO system taking into account the dynamic efficiency of a power amplifier (PA)  and adaptive activation of antennas following the DLP.

 Recently, both  massive MIMO and EE optimization of wireless systems have gained significant traction from both academia and industry \cite{Da, Hien,Yang,  Emilc, Emilj, Ngoc,  Magnus, Persson, Cheng}. In~\cite{Hien}, it has been shown that the transmit power can be reduced proportionally with the number of antennas if the BS has perfect channel state information (CSI) without any fundamental loss in performance. In \cite{Yang} and \cite{Emilc}, the impact of the circuit power on the EE of massive MIMO has been analyzed. Specifically, in \cite{Emilc}, it  has been shown that without accounting for circuit power consumption, an infinite EE can be achieved as the number of antennas, $M\to\infty$, which is misleading. In \cite{Emilc} and \cite{Emilj}, the authors show that the EE optimal strategy requires  increasing the transmission power with the number of antennas if the circuit power consumption is taken into account. They also show that the EE is a quasi-concave function of  the three main design parameters; namely, number of BS antennas, number of users and transmit power.  In \cite{Ngoc}, an adaptive antenna selection scheme was proposed where both the number of active RF (radio frequency) chains and the antenna indices are selected depending on the channel conditions.  In \cite{Persson}, it is shown that under  simultaneous per-antenna radiated power constraints and total consumed power constraints, the capacity achieving power allocation scheme feeds as many antennas as possible with maximum power. In \cite{Cheng}, the authors show that during low traffic demand, e.g., late night,  it is optimal to turn off a fraction of the antennas while minimizing total power consumption if the power consumption in the transceiver circuits is taken into account along with the RF amplifiers.  However, none of these studies has provided any mechanism to cope with the daily load variation and maintain high EE throughout the day in a multi-cell scenario. In  \cite{WS} and \cite{WS14}, we show that there is  potential to improve EE by adapting the number of antennas to the daily load variation.  Note that in \cite{WS14}, the number of active antennas that maximize EE for any given load has been derived considering fixed inter-cell interference i.e. the BS transmit power is assumed to be fixed for any number of active antennas and users.

In this extended work, we aim at designing a multi-cell massive MIMO system where each cell adapts the number of active antennas to the temporal variation of load following a DLP in order to maximize average EE throughout the day. In particular, we divide a day in an arbitrary number of intervals and for each interval we find i) the user distribution using the corresponding average load from the DLP and ii) the optimum number of antennas that maximizes EE under the given network load and interference condition for any number of users a cell serves.  We assume fixed average transmit power per antenna and also consider the realistic efficiency characteristics of non-ideal power amplifier (PA) when maximizing the EE.   Note that EE is measured in bit/Joule, i.e., the ratio between the average achievable data rate and the total average power consumption~\cite{Emilc, Chen}.  In order to map the user distribution to the DLP, i.e., to calculate the probability of having a certain number of users at a specific time of the day, we model the load at each BS as a state dependent $M/G/m/m$ queue~\cite{Cruz, Cheah} and  utilize the DLP as suggested in~\cite{EarthModel,CLRL}.

In order to compare the performance of this load adaptive system, we  dimension a reference system where the BSs do not adapt the number of antennas to the load variation, i.e., each BS uses the maximum number of antennas, $M_{\max}$ as long there is users to serve and achieves the highest EE when  serving the maximum number of users, $K_{\max}$ with $2\%$ blocking probability. We  find this $K_{\max}$,  $M_{\max}$ and the optimum average power per antenna, $p$ that maximize EE  and consider the amount of data carried by the cell at this state to be the peak load corresponding to the DLP. In order to  adapt the number of antennas to the DLP, it is important to note that the number of antennas that maximizes the EE of a cell depends on the number of antennas used by the other neighboring cells due to inter-cell interference. As the joint optimization problem is not convex with the number of antennas,  we propose a distributed algorithm that allows each cell to determine the number of active antennas through a best response iteration.  We resort to a game theoretic approach to show the convergence of the proposed algorithm to a unique Nash equilibrium.

 A traditional PA (TPA) reaches its maximum efficiency $\eta$ while operating at  its maximum output power, $P_{\max,\text{PA}}$. However, due to the non-constant envelope signals, e.g., OFDM, CDMA,  the PAs rarely operate at their maximum output power. Usually the peak to average power ratio (PAPR) of these signals is around $8$  dB~\cite{wegener}.  In this work we consider  TPA along with  a more efficient PA,  envelope tracking PA (ET-PA) in order to  investigate the impact of PA efficiency on our EE optimization scheme~\cite{Hossain1,Hossain}.

In  the numerical section, we illustrate the potential to increase EE by adapting the number of antennas in a multi-cell massive MIMO system.  We observe that the gain in EE can be as high as 250\% at very low load when compared with a baseline system that does not adapt the number of antennas to the DLP.  However, this gain  keeps decreasing with the increase in load. We see that over $24$ hour operation, the average gain is around 24\%.  In terms of actual power consumption, our scheme saves around 40\% energy  compared to the reference system over $24$ hour operation.
 
The rest of the paper is organized as follows: 
In   Section \ref{sec:ProblemFormulation}, we formulate the EE maximization problem and in  Section~\ref{sec:SystemModel}, we present required models and assumptions. In  Section \ref{sec:EEMG}, we formulate the EE maximization game and  in Section \ref{sec:OptimizationAlgorithm} we present the optimization algorithm. In Section \ref{sec:NumericalResults}, we illustrate the findings of the numerical analysis. We conclude the paper in Section \ref{sec:Conclusion}.

\section{Problem formulation}
\label{sec:ProblemFormulation}

In this study, we aim at designing a multi-cell massive MIMO system that adapts the number of active antennas to the network load and interference condition in order to maximize average EE.  The EE is the benefit-cost ratio of a network and is defined as the number of bits transferred per Joule of energy and hence can be computed as the ratio of average sum rate (in bit/second) and the average total power consumption (in Joule/second)~\cite{Emilj}. Power consumption of a BS consists of a fixed part (e.g., control signal, backhaul, DC-DC conversion) and a dynamic part, i.e., radiated transmit power and circuit power which depends on the number of active antennas and number of users served simultaneously.  If  $P_c^{tot}( K_c, M_c)$ denotes the total power consumed by an arbitrary BS $c$ when serving $K_c$ number of users simultaneously using $M_c$ antennas and  $R_c(K_c, M_c,  \{M_d\}_{d \neq c})$ denotes the resulting average data rate per user,  the corresponding EE will be
\begin{eqnarray}
\label{EE_c}
\mathrm{EE} &=&  \frac{{\text{Average sum rate}}}{\text{Power consumption}} = \frac {K_c R_c(K_c, M_c,  \{M_d\}_{d \neq c})} {P_c^{tot}(K_c, M_c)} 
\end{eqnarray} 
where $\{M_d\}_{d \neq c}$ represents the number of the antennas used by any other cell $d$, $d \ne c$.
The EE maximization problem for cell $c$ for a particular load can be expressed in the following way:
\begin{subequations} 
\label{eq:pproblem1}
\renewcommand\theequation{\theparentequation\roman{equation}}
\begin{align} 
   {\mathop {{\text{maximize}}}_{M_c} } &   
    \frac {K_c R_c(K_c,  M_c, \{M_d\}_{d \neq c}) } {P_c^{tot}(K_c, M_c)}\label{eq:P11}\\ 
 {{\text{subject to:}}} &\;\; { M_c \in \mathcal{S}_c}  \label{eq:P12}  
\end{align}
\end{subequations}
where $\mathcal{S}_c =\{0, 1, 2, . . ., M_{\max}\}$ is the set of feasible number of antennas the cell $c$ can use. 
However, the network loads vary throughout the day. In order to capture the daily load variation and maximize EE throughout the day, we model the load at each BS as an $M/G/m/m$ state-dependent queue.  Let us consider that  during the time interval $h$, the steady state probability of the BS $c$ serving $n$ number of users, i.e., $Pr[K_c=n]$  is denoted by $\pi_c(h,n)$. Our objective is to maximize the average EE by adapting the number of active antennas to the number of active users taking the received interference into account. The main problem formulation for BS $c$ can be rewritten as
\begin{subequations} 
\label{eq:problem}
\renewcommand\theequation{\theparentequation\roman{equation}}
\begin{align} 
   {\mathop {{\text{maximize}}}_{\mathbb{M}_c}} &\;\;
   {\sum_{h=1}^{H}}
     \sum_{n=1}^m \!\!\pi_c(h, n) \frac {n R_c(n, M_c^{(h)}, \{M_d^{(h)}\}_{d \neq c}) } {P_c^{tot}(n, M_c^{(h)})}  \label{eq:Q11}\\ 
   {{\text{subject to:}}} &\;\; { \mathcal{M}_c^{(h)} \in \mathcal{S}_c} \label{eq:Q12}  
\end{align}
\end{subequations}
where $R_c(n, M_c^{(h)}, \{M_d^{(h)}\}_{d \neq c}) $ is the average rate per user when there are $n$ users in the cell during time interval $h$. ${\mathbb{M}_c} = [\mathcal{M}_c^{(1)} \mathcal{M}_c^{(2)} . . .  \mathcal{M}_c^{(H)} ]$ where $\mathcal{M}_c^{(h)}$ is the vector that gives the number of antennas that maximizes the EE at different user states in cell $c$ during the time interval $h$. It is important to note that the optimum number of antennas is not only dependent on the number of users but also the time of the day which determines the probability of having that certain number of users. Note that EE optimization could be divided into an arbitrary number of time intervals, $H$, during a daily traffic profile. 
In this work, we use  hourly average load (i.e., $H=24$) from the DLP proposed in \cite{EarthModel} and  $12$ minute average load (i.e., $H=120$) from the DLP proposed in  \cite{CLRL} as input for the simulations in order to optimize the system over $24$ hour operation.  As the solution for  the problem \eqref{eq:problem}  depends on the actions taken at other BSs, we formulate the joint optimization  under a game theory framework  where each BS decides the  number of active antennas in a distributed manner, iteratively converging to a Nash equilibrium. Note that the formulated problem is not dependent on this particular queuing model.  The queuing model serves for the mapping of a daily traffic profile in terms of arrival rates to the expected values of number of users at a certain time interval of the day.

\section{Models and Assumptions}
\label{sec:SystemModel}

In order to study the energy efficient traffic adaptive massive MIMO system,  let us consider the following  models for transmission rate, traffic
and BS power consumption.

\subsection{System model}
Let us consider the  downlink of a multi-cell massive MIMO system consisting of cells with indices in the set $\mathcal{C} = \{1, 2, ..., C\}$ and each having its own BS. In the following the terms cell and BS are used interchangeably. The BS $c \in \mathcal{C}$  uses $M_c$ antennas to  serve $K_c$  single antenna UEs. The set of user states for a cell $c$ is $\mathcal{U}_c = \{1, 2,  ..., K_{\max}\}$,  with the elements representing the number of users served in a cell. Each antenna of the BS has its own power amplifier. We consider Rayleigh fading channels to the UEs and the spacing between adjacent antennas at the BS is such that the   channel components  between the BS antennas and the single-antenna UEs are uncorrelated. Under the assumption of this independent fading and considering the fact that power gets averaged over many subcarriers,  each antenna uses the same average power. Let us denote  this average transmit power per antenna by $p$, hence, the total transmit power of cell $c$ is  $p M_c$. The number of active antennas of any other cell $d \ne c$ is $M_d$ and the corresponding transmit power is $p M_d$. Note that the average transmit power of a BS is not fixed as it varies with the number of active antennas.\footnote{ We assume the transmit power is limited by the power amplifiers and not by any regulatory aspects.}  The large-scale fading, i.e., the average channel attenuation due to path-loss, scattering, and shadowing from the BS to the UEs is assumed to be the same for all the antennas, as the distance between any UE and the BS is  much larger than the distance among the antennas. 

Let us assume that the BSs and UEs employ a time-division duplex (TDD) protocol and are perfectly synchronized. We also assume that the BS obtains perfect CSI from the uplink pilots, which is a reasonable assumption for low-mobility scenarios. Each BS employs zero forcing precoding so that the intracell interference is canceled out and  equal power allocation is applied to utilize the pathloss differences to achieve high average data rates. Let us denote the average data  rate of the users  in cell $c$ by $R_c$. Note that $R_c$ is a function of $K_c$, $M_c$ and $M_d$, $d \in \mathcal{C}$, $d \ne c$. The average data rate is limited by the capacity, which is generally unknown in multi-cell scenarios, but a tractable and achievable lower bound in cell $c$ under the above assumptions is  given by \cite{ Emilmo}  
\begin{equation}
R_c\!=\! B\Big(1\!-\!\frac{ \alpha K_{\max}}{T_c}\Big) \log_{2}\! \Bigg(\!1\!+\!\frac{p \frac{M_c}{K_c}  (M_c-K_c)}{\sigma^2 G_{cc} + \sum_{d \neq c}  p M_d G_{cd}} \!\Bigg)
\label{eq:Rc}
\end{equation} 
where $\alpha$ is the pilot reuse factor,  $K_{\max}$ is the maximum number of users  assumed to be the same for all cells (i.e., the number of pilot sequences = $\alpha K_{\max}$), $T_c$ is the length of the channel coherence interval (in symbols), $\Big(1-\frac{\alpha K_{\max}}{T_c}\Big)$ accounts for the necessary overhead for channel estimation,  $B$ is the effective channel bandwidth,  $G_{cc} = \mathbb{E}\{\frac{1}{g_{cck}}\}$ and $G_{cd}$=$\mathbb{E}\{\frac{g_{dck}}{g_{cck}}\}$, where  the random variable  $g_{cck}$ is the channel variance from the serving BS and $g_{dck}$ is the  channel variance from cell $d$ to users in cell $c$, i.e.,  $ \sum_{d \neq c} G_{cd} p M_d$ is the average inter-cell interference power from cell $d$ to $c$  normalized by $G_{cc}$. Note that $R_c$ is achieved by averaging over the locations of the users of cell $c$, thus it can either be viewed as an average rate among the users in the cell or a rate that every user will achieve when it is moving around in the cell. A proof of the rate formula in \eqref{eq:Rc} is provided in Appendix \ref{RateFormulaProof}. Note that the optimization framework developed in this paper can be used along with other rate formulas as well.

\subsection{Traffic model} 
 \label{Sec:TrafficModel}
In a BS with massive MIMO, a certain number of users are served simultaneously using  a certain number of active antennas. The rate achieved by each user depends on the number of users served simultaneously and the number of antennas that are active. As shown in~\cite{Emilj},  there is a particular combination of number of active antennas and simultaneously served users that yields the maximum EE.  In this study, this optimum number of antennas  and users are denoted by  $M_{\max}$ and  $K_{\max}$ respectively and are the upper limit for the number of active antennas and simultaneously served users. The $K_{\max}$ users that can be served simultaneously is derived by modeling the load at each BS as a state-dependent $M/G/m/m$ queue. The $M/G/m/m$ queue dictates that for exponential arrival and general distribution of service time, maximum $m$ number of users can be served simultaneously (number of servers = $m$, i.e., $K_{\max}$, waiting place = 0). The state dependency arises from the fact that the user rate depends on the number of users the BS serves simultaneously. The network is assumed to be  dimensioned in a way that the data carried by a cell while serving $K_{\max}$ users  corresponds to the peak load of the DLP. In a queuing system with no buffer space, the blocking probability is equal to the probability of having the system 100\% loaded, i.e., probability of having $K_{max}$ number of users.  This can be explained by the PASTA (Poisson Arrivals See Time Averages) property which holds when the arrivals are following the Poisson process\cite{PASTA}.  Therefore, $2$\% blocking probability at the peak load  means  the probability of serving  $K_{\max}$ users simultaneously is $0.02$.  In order to capture the daily traffic variations, we consider the  DLP proposed for data traffic in Europe  \cite{EarthModel} and DLP proposed for residential and business areas in \cite{CLRL}.   The steady state  probabilities for the random number of users in the BS $c$, $\pi_c(n)\equiv Pr[K_c\!=\!n]$, are as follows \cite{ Cruz}
\begin{equation}
\pi_c(n)\!=\!\!\left[\frac{\left[\lambda \frac{s}{R_c(1)}\right]^n}{n! f(n)f(n-1)...f(2)f(1)} \right]\pi_c(0), n \in \mathcal{U}_c, 
\end{equation}
where $\pi_c(0)$ is the probability that there is no user in cell $c$ and is given by
\begin{equation}
\pi_c^{-1}(0)=1+ \sum_{i=1}^{K_{\max}}\left(\frac{\left[\lambda\frac{s}{R_c(1)}\right]^i}{i! f(i)f(i-1)... f(2) f(1)}\right)  
\end{equation}
 \begin{figure}[!tp] 
\centering
        \includegraphics[ scale=0.8]{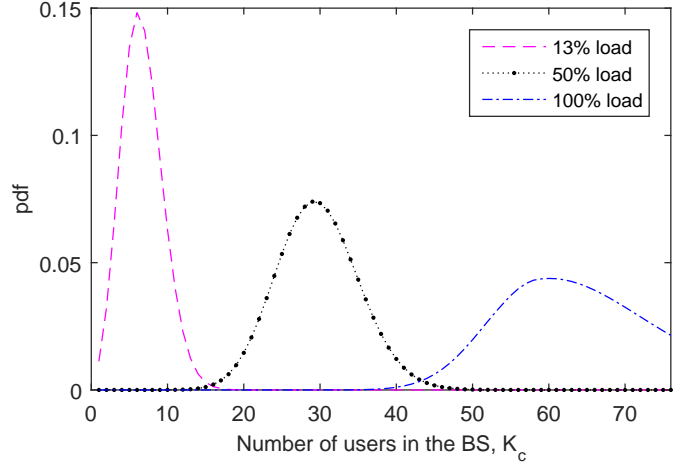}	
\caption{	 Distribution of active users in the cell  at 13\%, 50\%  and 100\%  load with the parameters provided in Section \ref{sec:NumericalResults}. } 
\label{fig:pdf}
\end{figure}
where $\frac{s}{R_c(1)}$ is the expected service time when BS $c$ serves  a single user, $f(n)=R_c(n)/R_c(1)$ where $R_c(n)$
 is the  data rate per user while serving $n$ number of users, $\lambda$ is the arrival rate, and $s$ is the total data traffic contribution by a single user. Note that we use (\ref{eq:Rc}) to find the data rates at different user states.  In order to find the steady state probability distribution throughout the day, first we set the values for $\lambda$ and $s$. As we allow 2\% blocking probability  while serving $100$\% load, we find the maximum $\lambda$, i.e., $\lambda_{\max}$ that gives $\pi(K_{\max}) = 0.02$ for a fixed $s$. Assuming that $s$ remains constant, we derive the average number of users for other time intervals following the DLP using $\lambda_{\max}$. For example, from the DLP, if the average load at any time interval $h$ is  $x$\%, the corresponding average number of users is $\lambda_h = \frac{x}{100}\cdot \lambda_{\max}$. This  $\lambda_h$ has been used as the input to the $M/G/m/m$ queue to find the steady state probability distribution  of the users during time interval $h$. Fig. \ref{fig:pdf} gives three example plots of the user distribution for 13\%, 50\% and 100\% loads with the parameters provided in Section \ref{sec:NumericalResults}. Note that for 100\% load, the probability of serving $K_{\max} = 76$ users is $0.02$.

\subsection{Power consumption model}
The total power consumed in a BS $c$ with $M_c$ active antennas and $K_c$ users is modeled as
\begin{equation} 
P_c^{tot}(K_c, M_c)= M_c P_{PA} (p )+P_\mathrm{BB} (K_c, M_c)+P_\mathrm{Oth} 
\end{equation}
where $P_{PA}(p)$ gives the power consumption of a PA when the average output power is $p$, $P_\mathrm{BB} (K_c, M_c)$ is the base band signal processing power when the BS serves $K_c$ number of users simultaneously with $M_c$ number of antennas. $P_\mathrm{Oth}$ includes the load-independent power for site cooling, control signal, DC-DC conversion loss, etc.

\subsubsection{Baseband power consumption}
For baseband and fixed power consumption we use the model proposed in \cite{Emilc}. The total circuit power is given by 		
\begin{equation} 
\label{eq:Pcp}
P_\mathrm{CP}=P_\mathrm{TC}+P_\mathrm{CE}+P_\mathrm{C/D}+P_\mathrm{LP}. 	
\end{equation}				
The power consumed in the transceiver is given by $P_\mathrm{TC}=M P_\mathrm{BS}+P_\mathrm{SYN}$  where $P_\mathrm{BS}$ is the power required to run the circuit components,  e.g., converters, mixers and filters attached to each antenna at the BS and $P_\mathrm{SYN}$ is the power consumed by the local oscillator. $P_\mathrm{CE}$ is the power required for channel estimation. $P_\mathrm{C/D}$ is the total power required for channel encoding, $P_\mathrm{COD}$,  and channel decoding, $P_\mathrm{DEC}$.   $P_\mathrm{LP}$ is the power consumed by linear processing.
According to \cite{Emilc}, the total baseband power can be expressed as 
\begin{equation} 
\label{eq:TPbb}
P_\mathrm{BB} (M_c, K_c) = \underbrace {\mathcal{A} K_c R_c+\sum_{i=0}^3 C_{0,i} K_c^i}_{C_0^{BB}}+ M_c \underbrace{\sum_{i=0}^2 C_{1,i} K_c^i}_{C_1^{BB}}  
 \end{equation}
where  $C_{0,0}=P_\mathrm{SYN},  C_{0,1}=0,  C_{0,2}=0, C_{0,3}=\frac{B}{3 T_c L_\mathrm{BS}}, C_{1,0}= P_\mathrm{BS}, C_{1,1}=\frac{B}{L_\mathrm{BS}}(2+\frac{1}{T_c}), C_{1,2}=\frac{3B}{L_\mathrm{BS}}, \mathcal{A}=P_\mathrm{COD}+P_\mathrm{DEC}$, $R_c$ is the average user rate  as given in  (\ref{eq:Rc}) and $B$ is the  bandwidth. The power consumption $C_1^{BB}$ gets multiplied with $M_c$ whereas $C_0^{BB} $ does not. 
 
\subsubsection{Power amplifier}
We consider both TPA and ET-PA in this analysis. In case of the TPA, the maximum efficiency is achieved only when it operates  near the compression region, i.e., at maximum output power. On the other hand, ET-PA achieves better efficiency throughout the operating range by tracking the radio frequency (RF) signal and regulating the supply voltage accordingly. In order to  transmit mean output power $p$, the total input power needed by TPA  can be approximated as~\cite{Persson, Hossain} 
\begin{equation}\label{eq:T-PA}
P_{TPA}(p)\approx \frac{1}{\eta} \sqrt{p\cdot P_{\max,\text{PA}}}
\end{equation}
where $\eta$ denotes the maximum PA efficiency when transmitting maximum output power, $P_{\max,\text{PA}}$. Note  that, $p$ reflects the average power and the PA must be able to handle the peak power, i.e., $p$ must be at least $8$ dB less than $P_{\max,\text{PA}}$ as the typical  PAPR of recent technologies, e.g., OFDM is $8$ dB. 

In case of ET-PA, the total power consumption is approximately a linear function of the actual transmit
power~\cite{Hossain},
\begin{equation} \label{eq:ET-PA}
P_{ET-PA}(p)\approx \frac{1}{(1+\epsilon)\eta}\left(p+\epsilon P_{\max,\text{PA}}\right)
\end{equation}
where $\epsilon \approx 0.0082$ is a PA dependent parameter. 
The total power required by $M_c$ antennas in order to deliver output transmit power $p$ per antenna is  $ C_1^{PA}M_c$ where 
\begin{equation}
C_1^{PA} = \begin{cases}  \frac{1}{\eta} \sqrt{p P_{\max,\text{PA}}} , & \mbox{if } \mbox{ TPA} \\ \frac{p}{(1+\epsilon)\eta}+ \frac{\epsilon P_{\max,\text{PA}}}{(1+\epsilon)\eta}, & \mbox{if } \mbox{ ET-PA} \end{cases}
\end{equation} 
Since part of the power consumption increases approximately linearly with the number of antennas and some other parts do not, we can write total power as 
\begin{equation}
\label{eq:ETPAeff}
P_{total} \approx C_0+C_1 M_c 
\end{equation}
 where $C_0=C_0^{BB}+P_{Oth}$ and $C_1=C_1^{PA}+C_1^{BB}$.
\section{EE maximization game} \label{sec:EEMG}
In this section, we utilize the models that were introduced in the previous section to formulate the EE optimization problem under a game theory framework. As the transmit power of a BS is proportional to the number of active antennas and depends on the transmit power of interfering BSs due to inter-cell interference,  the number of active antennas for different cells are coupled.  As a result, the objective function in \eqref{eq:problem} for cell $c$ is dependent on the number of antennas used by other cells and there is no closed form expression for $M_c$ that would maximize EE. The objective function in ~\eqref{eq:problem} involves summation over all the user states at each time interval. We drop $(h)$  as the optimization for different time interval $h$ can be carried out separately following the same procedure and   rewrite the problem (\ref{eq:problem}) using  \eqref{eq:Rc} and \eqref{eq:ETPAeff} for a single time interval as

\begin{subequations} 
\label{eq:ProblemLong}
\renewcommand\theequation{\theparentequation\roman{equation}}
\begin{align} 
   {\mathop {{\text{ maximize}}}_{\mathcal{M}_c}} &\;\;   
     \sum_{n=1}^{K_{max}} \!\!\pi_c(h, n) \frac{n \beta \log{({1-n M_c \gamma_{c,1}+\gamma_{c,1}M_c^2})}}{C_0+C_1 M_c} \label{eq:P21}\\
    {{\text{subject to}}} &\;\; { n+1<M_c(n)< M_{\max}}   \label{eq:P22}  
\end{align}
\end{subequations}
where  $\gamma_{c,1} =\frac{\frac{1}{n}p}{(G_{cc} \sigma^2+ \sum_{d \neq c} G_{cd} p M_d )}$ is the achieved signal to interference plus noise ratio (SINR) in cell $c$ when using a single antenna and $\beta = (1-\frac{\alpha K_{\max}}{T_c}) \frac{B}{\text{ln}2}$. Furthermore,  when the BS is serving a particular number of users, $n$, the objective function  can be  broadly written as
\begin{equation}
E_c  \approx \frac  {n \beta\log{({1-n M_c \gamma_{c,1}+\gamma_{c,1}M_c^2})}} {C_0+C_1 M_c}. 
\label{eq:Ebfinal}
\end{equation}
The optimization  problem \eqref{eq:Ebfinal} is not concave in $M_c$ and $M_d$ for either  TPA or ET-PA. However, this objective function has a nice property of increasing differences in $M_c$ and $M_d, \quad \! \! \forall d\in \mathcal{C}, d\ne c,$ in the sense of Topkis ~\cite{Topkis}. Because of that we resort to an algorithm based on best response iteration in a game theoretic framework. In this framework, each BS  iteratively  finds the most energy efficient number of active antennas taking into account the interference from the surrounding BSs.

In order to formulate \eqref{eq:ProblemLong} in a game theoretic framework, let us define  the strategy space, $\mathcal{S}_c$, which contains the set of  feasible number of antennas used by any cell $c$ at different user states,  $ n \in\mathcal{U}_c$. Note that the set $\mathcal{S}_c$ is a function of  ${\mathbf{M}}_{-c}$, i.e., the  number of antennas used by all other cells that interfere cell  $c$  and can be written as  
\begin{equation}
\label{eq:strategy}
\mathcal{S}_c(\mathbf{M}_{-c})\!=\!\left\{\mathcal{M}_c(n) \!: \!n\!+\!1\!\! \leq \!  \mathcal{M}_c(n) \!\leq\! M_{\max},  \forall n\in\mathcal{U}_c\right\} 
\end{equation} 
where the lower bound comes from the constraint for zero forcing. 

Next, we  define the {\em EE maximization game}, $\mathcal{G}(\mathcal{K},\mathcal{S},\mathcal{E})$ where the players are the BSs,  $S=S_1 \times S_2 \times \cdots \times S_C$ is the strategy space (i.e., space of number of active antennas), and ${\mathcal{E}}=E_c(\mathcal{M}_{c},\mathbf{M}_{-c}), c\in\mathcal{C}$ is the utility of the players (i.e., EE  of the different cells). 
 The {\em best response} is the strategy (or strategies) that produces the most favorable outcome for a player given other players' strategies. The use of the best response strategy gives rise to a dynamic system of the form
\begin{equation}
\label{eq:maximize}
\mathcal{M}_c={\arg\max}_{\mathcal{M}_c\in\mathcal{S}_c(\mathbf{M}_{-c})} 
{E}_{c}(\mathcal{M}_c,\mathbf{M}_{-c}), \forall c. 
\end{equation}

In order to study the convergence  properties of the best response strategy of dynamically adapting the number of antennas, let us first look at the following definitions available in ~\cite{Levin, Altman, Topkis}. 
\begin{definition}
The strategy vector $\mathcal{M}_c^*$  is said to be a Nash equilibrium of the game  $\mathcal{G}(\mathcal{K},\mathcal{S},\mathcal{E})$, if no player can unilaterally increase its utility function, i.e., 
\[
{E}_{c}(\mathcal{M}_{c}^*,\mathbf{M}_{-c}^*)\ge
{E}_{c}(\mathcal{M}_{c},\mathbf{M}_{-c}^*), \forall 
c\in \mathcal{C}, \mathcal{M}_c\in \mathcal{S}_c
\]
where $\mathbf{M}_{-c}^*$ is the vector of antennas used by the cells other than $c$ at Nash equilibrium.
\end{definition} 
\begin{definition}
Let a function $\psi:\mathcal{X}\times \mathcal{Y} \to \mathbb{R}$ be continuous and twice differentiable. Then, the following two statements are equivalent
\begin{itemize}
\item $\frac{\partial^2 \psi(x,y)}{\partial x\partial y }\geq 0, 
\forall y \in \mathcal{Y}, \forall x \in \mathcal{X}$. 
\item The function $\psi$ has increasing differences in $(x,y)$ i.e. 
\[
\psi(x',y')-\psi(x,y')\geq  \psi(x',y)-\psi(x,y), 
\forall (x',y') 
\]
where $x'\geq x, y'\geq y$.
\end{itemize}
\end{definition}

The function $\psi(x,y)$ does not have to be nicely behaving,  nor does $\mathcal{X}$ and $\mathcal{Y}$ have to be intervals.
\begin{definition}
The  EE maximization game, $\mathcal{G}(\mathcal{K},\mathcal{S},\mathcal{E})$, is said to be supermodular if
\begin{itemize}
\item the strategy space $\mathcal{S}_c, \forall c$ is a compact subset of $\mathbb{R}$;
\item the utility function ${E}_{c}, \forall c$ is upper semi-continuous in $\mathcal{M}_c$, $\mathbf{M}_{-c}$;
\item the utility function ${E}_{c}$,  has increasing differences in $\mathcal{M}_c$, $\mathbf{M}_{-c}$.
\end{itemize}
\end{definition}

 \begin{proposition}
 \label{SupermodularGame}
 In a supermodular game, a Nash equilibrium exists. Also, if we start with a feasible policy $\mathcal{M}$, then the sequence of best response iterations monotonically increases in all components and converges to a Nash equilibrium. 
 \end{proposition}   
 \begin{proof}
 These properties of the supermodular game follow from~\cite[Theorem 1]{Altman}.
 \end{proof}
In order to show that the best response iteration strategy~\eqref{eq:maximize} leads to a Nash equilibrium, it is sufficient to show that the {\em EE maximization game} $\mathcal{G}(\mathcal{K},\mathcal{S},\mathcal{E})$ is supermodular. 
   
\begin{proposition}    
{\em EE game} $\mathcal{G}(\mathcal{K},\mathcal{S},\mathcal{E})$ is supermodular.
\end{proposition} 
\begin{proof}
 The {\em EE game} $\mathcal{G}(\mathcal{K},\mathcal{S},\mathcal{E})$ is proved to be supermodular in Appendix \ref{SupermodularProof} for the rate model described in Section \ref{sec:SystemModel}.
\end{proof}

\section{Optimization algorithm} 
\label{sec:OptimizationAlgorithm}
 In this section, we  describe the best response optimization algorithm  and present it in the form of pseudocode, see Algorithm $1$.  Initially, the number of antennas of all the cells are set  to the maximum, $M_{\max}$. We start the best response iterations from the $c$-th cell. The average interference level received by any user in the cell  $c$, $I_c({\mathbf{M}}_{-c})$, is a function of the number of antennas used in the other cells, $\mathbf{M}_{-c}$. In order to approximate the interference level  $I_c(\mathbf{M}_{-c})$, it is assumed that the $d$-th interfering BS, $d\neq c$, transmits with the number of antennas found from the weighted mean of the number of antennas for its different user states i.e.,  $\sum_{n=1}^m{\mathcal{M}_d(n) \pi_d(n)},  \forall d\!:\!d\neq c$. Once the interference caused to the  $c$-th cell is calculated, we identify the strategy space for the $c$-th cell based on equation~\eqref{eq:strategy}.  Finally, we find the vector of antennas that maximizes EE at different user states,  $\mathcal{M'}_c \leftarrow {\arg\max}_{\mathcal{M}_c\in\mathcal{S}_c(\mathbf{M}_{-c})} E_c(\mathbf{M}_{-c})$ where '$\leftarrow$' indicates the direction of value assignment. We  iterate over all the cells and  optimize the antenna vector for each cell. The iterations are carried out until the antenna vector for each cell converges, i.e., there is one iteration where none of the antenna numbers changes. 

 \begin{algorithm}
\label{algo}
\caption{Best response iteration}
\begin{algorithmic}[!t]
\STATE ${\mathcal{M}}_{c} \leftarrow M_{\max} \cdot \mathbf{1}, \forall c\in\cal{C}$. 
\STATE $\textnormal{maxtol} \leftarrow 1 $
\WHILE{$\textnormal{maxtol} \neq 0$} 
\FORALL{$c\in\cal{C}$}
\STATE ${\mathbf{M}}_{-c}\leftarrow \sum_n{\mathcal{M}_d(n)\pi_n}, \forall n\in {\cal{U}}_{d},\forall d: d \neq c$ 
\STATE Define strategy space $\mathcal{S}_c$ based on~\eqref{eq:strategy}
\STATE $\mathcal{M'}_c \leftarrow {\arg\max}_{\mathcal{M}_c\in\mathcal{U}_c(\mathbf{M}_{-c})}
{E}_{c}(\mathbf{M}_{-c})$
\STATE $\mathrm{tol}_c\leftarrow ||{\mathcal{M'}}_{c} - {\mathcal{M}}_{c}||_{0}$\footnotemark
\STATE $\mathcal{M}_c \leftarrow \mathcal{M'}_c$
\ENDFOR
\STATE $\textnormal{maxtol} \leftarrow {\text{max}}_{c}(\mathrm{tol}_c) $
\ENDWHILE
\end{algorithmic}
\end{algorithm}

The number of antennas for a user state $n$,  $M_c(n),  n\in {\mathcal{U}}_c$, is independent of the antennas used  for the other user states in the same cell. Therefore, in order to carry out the optimization step, $\mathcal{M'}_c \leftarrow {\arg\max}_{\mathcal{M}_c\in\mathcal{S}_c(\mathbf{M}_{-c})} {E}_{c}(\mathbf{M}_{-c})$, it is sufficient to solve the optimization problem~\eqref{eq:ProblemLong} separately for each user state.  For the given number of antennas of the interfering BSs, the interference is known. For known interference the EE problem for any user state is a quasi-concave function of $M_c$ as it is a ratio of a concave and an affine function of $M_c$ \cite{Ratio}. As a result, it suffices to compute the value of the objective function at the stationary point and at the end points of the interval to identify the optimal $M_c$. 

\footnotetext{ $||\cdot||_{0}$ is the zero norm which gives the number of nonzero elements in the set.}
 
\section{Numerical analysis}
\label{sec:NumericalResults}
We consider the downlink of a cellular network with $19$ regular  hexagonal cells where the wrap around technique is applied in order to get rid of the boundary effect. The maximum cell radius is $500$ m if not otherwise specified. The users are uniformly distributed in the cell with a minimum distance of $35$ m from the cell center. Note that we consider $15000$ test points in each cell in order to calculate the average channel variance from the serving BS, $G_{cc}$ and the average inter-cell interference power, $G_{cd}$. Traffic is assumed to be homogeneously and independently distributed over the cells. As a result, the network load and cell load are the same and has been used interchangeably. 
\begin{table}[!ht]
\caption{Simulation parameters}
\label{RP} \centering
\begin{tabular}{|l|l|} \hline
\multicolumn {2} {|c|}{\textbf {Reference parameters }} \\
 \hline
  \textit{Parameter} & \textit{Value} \\
 \hline
 Number of cells & $19$ \\
 \hline
Grid size inside each cell & $15000$ points \\
 \hline
Cell radius: $d_{\max}$ & $500$ m \\
\hline
Minimum distance: $d_{\min}$ & $35$ m \\
 \hline
Maximum PA efficiency  at $P_{\max,\text{PA}}$ & $80$\%\\
 \hline
Path loss at distance $d$ & $\frac{10^{-3.53}}{||d||^{3.76}}$\\
 \hline
Local oscillator power: $P_\mathrm{SYN}$  & $2$ W\\
\hline
BS circuit power: $P_\mathrm{BS}$ & $1$ W \\
\hline
Other power: $P_\mathrm{Oth}$  & $18$ W\\
\hline
Power for data coding: $P_\mathrm{COD}$  & $0.1$ W/(Gbit/s) \\
\hline
Power for data decoding: $P_\mathrm{DEC}$  & $0.8$ W/(Gbit/s) \\
\hline
Computational efficiency at BSs:$L_\mathrm{BS}$ & 12.8 Gflops/W \\
\hline
Bandwidth  & $20$ MHz\\
 \hline
Total noise power: $B\sigma^2$  & -$96$ dBm\\
 \hline
Channel coherence interval: $T_c$  & $5000$ symbols \\
\hline
Pilot reuse factor: $\alpha$ & 7 \\
\hline
\end{tabular}
\end {table}
 However, the solution and methodology are still valid for non-homogeneous area traffic distribution scenario although energy saving values may slightly vary. In order to find the overall performance throughout the day we divide the day in 120 intervals, i.e.,  12 minute average load is used to find the user distribution. We use the DLP suggested for residential areas in \cite{CLRL}  for this purpose.  The parameters for the simulation are given in  Table \ref{RP}. Some of them are taken from~\cite{Emilj}.

\subsection{Reference system }
\label{globalAntennaUser}
In this work, we propose a distributed algorithm  for a multi-cell massive MIMO system where each cell maximizes average EE adaptively with the variation of network load. In order to compare the performance of the devised mechanism, we need to have a reference system. In our reference system the BSs do not adapt the number of antennas to the  variation of load, i.e., use $M_{max}$ if there is at least one user to serve and turns off all the antennas otherwise. Note that minimum network load is considered to be $10\%$ in order to account for the system information and control signaling.  We dimension our reference system in a  way so that it achieves highest EE when each BS  serves $K_{max}$ users with $M_{max}$ antennas, i.e., the network is most energy efficient when serving peak load according to the DLP.   From equation \eqref{eq:problem} and  \eqref{eq:Rc}, it can be observed that for a given interference, i.e., power used by other BSs, the EE of BS $c$ depends on  $M_c$ , $K_c$,  and $p$.  In order to find $K_{\max}$, $M_{\max}$  and optimum $p$ for the reference system, exhaustive search is used to search over different combinations of $M_c$ and $K_c$ and  optimal power allocation, $p$  is computed for each combination as $M_c$ and $K_c$ are integers and  the EE is a quasi-concave function of $M_c$, $K_c$.  In case of ET-PA, for a given  $M_c$ and  $K_c$, the EE is a quasi-concave function of  $p$ as it reduces to a ratio of a concave and an affine function\cite{Ratio}. However, for TPA, the function becomes the ratio of concave and convex function which can be maximized by the branch and bound algorithm proposed in \cite{Ratio2}. For the parameters given in Table \ref{RP}, the maximum EE is achieved when serving $76$ users with $158$ antennas where the transmit power of each antenna, $p$ is  $0.1$ Watt. Similarly, for ET-PA, $K_{\max}= 68$, $M_{\max}= 134$ and $p= 0.18$ Watt. Once the reference system is dimensioned, the number of antennas are kept fixed at $M_{\max}$ in order to find the reference EE and user rates for network loads other than the peak. EE is calculated by iterating over the activity factor of each BS starting from a given inter-cell interference and converging on the activity factor, defined as the probability of finding a BS active at a particular instant.  Note that at the peak load the activity of each BS is considered to be $1$.

\subsection{Interplay between number of active users and active antennas}The probability of getting different number of users served by the BS simultaneously depends on the average network load, for example, see Fig. \ref{fig:pdf}. In Fig. \ref{fig:UserVsAntenna}, we show how $M_c$ varies with $K_c$ when the  network loads are low (13\%), medium (50\%)  and high (100\%)  which corresponds to the  user distribution as shown in Fig. \ref{fig:pdf}.
\begin{figure}[!tbp]
  \centering
  \begin{minipage}[b]{0.49\textwidth}
    \includegraphics[width=\textwidth]{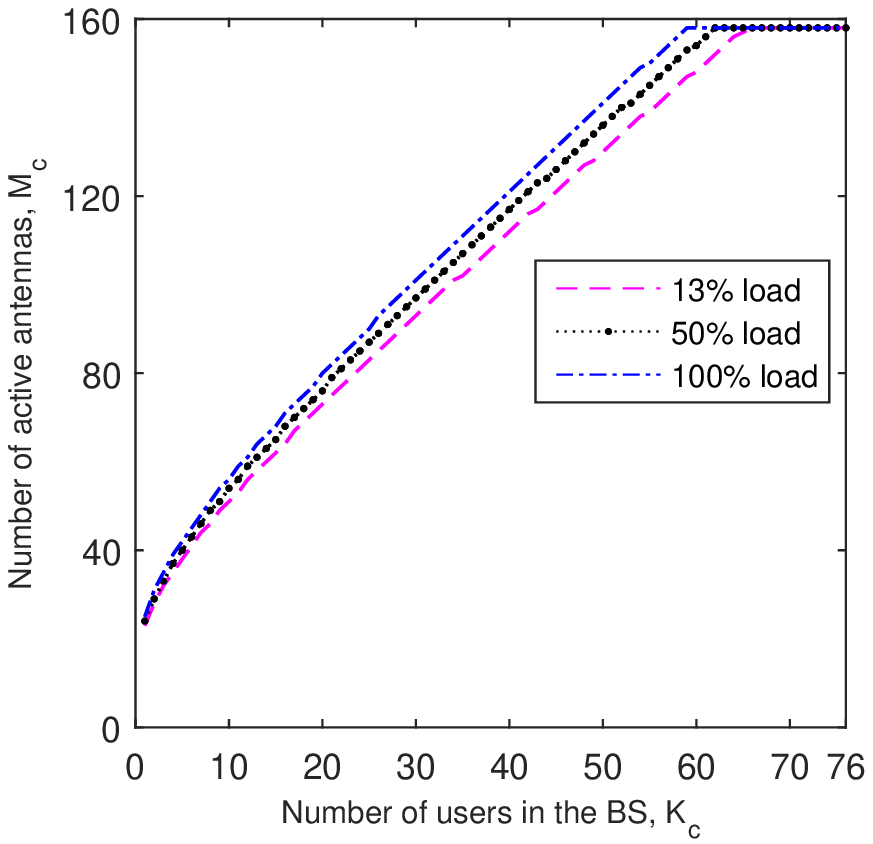}
    \caption{Number of antennas as a function of the number of users in the cell  at 13\%, 50\%  and 100\%  load.}
		\label{fig:UserVsAntenna}
  \end{minipage}
  \hfill
  \begin{minipage}[b]{0.49\textwidth}
    \includegraphics[width=\textwidth]{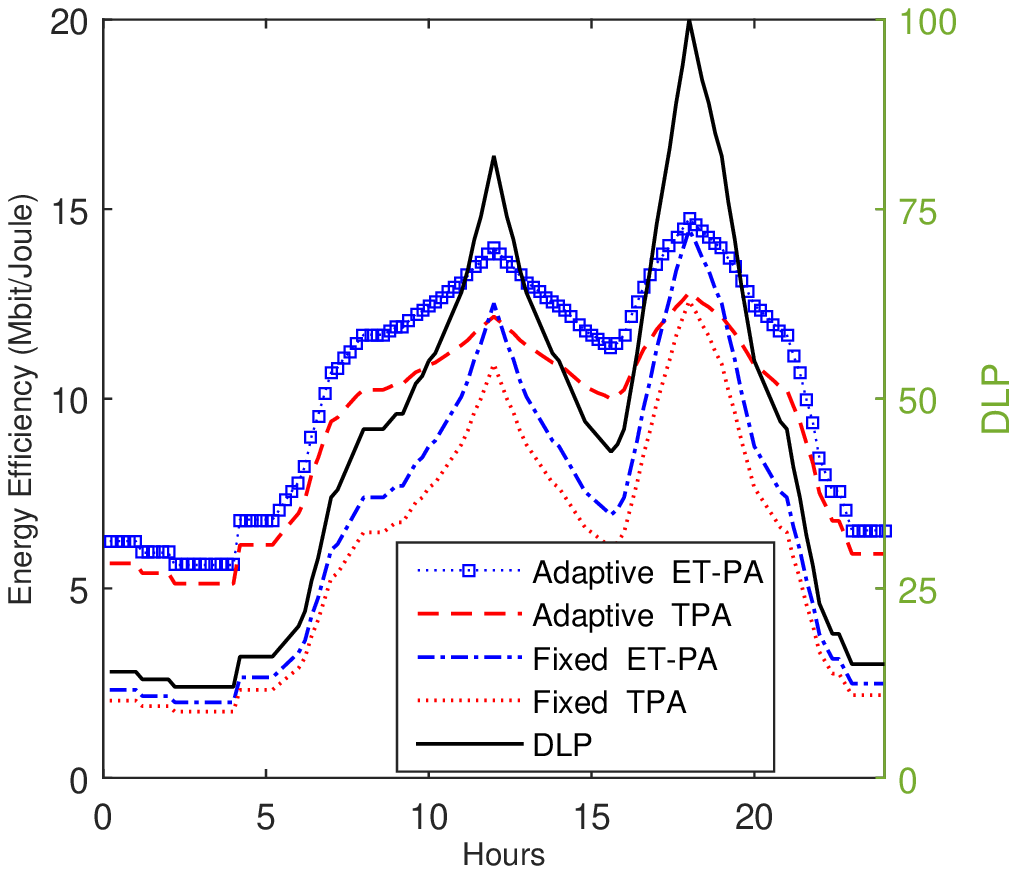}
    \caption{Average EE over 24 hour load variation for both fixed and adaptive antenna while using both TPA and ET-PA.}
		\label{fig:EEtPAeTPA}
  \end{minipage}
\end{figure}
 When the EE is optimized over the varying number of users, $M_c$ adapts to the user profile resulting in an approximately linear relation between $M_c$ and $K_c$. However, the ratio between $M_c$ and $K_c$ is quite high when the BS serves only a few users and ends up around two at high load. When $K_c$ is very small, the energy consumed by activation of an additional antenna is not significant compared to the fixed energy consumption but contributes significantly to increase the EE due to  higher array gain, i.e., $(M_c-K_c)$. Another observation from  Fig. \ref{fig:pdf} and Fig. \ref{fig:UserVsAntenna} is that although the  $M_c$ for a particular $K_c$ does not change a lot at different network load,  the average number of antennas used by a BS at different loads varies significantly due to the probability distribution of the users. Note that  the rest of the results are generated by taking weighted average over the performance of all the user states for each particular average network load.

\subsection{EE gain at different network load}
In  Fig. \ref{fig:EEtPAeTPA}, we show the hourly average EE throughout the day following the DLP for both the TPA and ET-PA.  In general, the EE increases with the increase in load for both  the reference case and our scheme for both TPA and ET-PA.  Our scheme achieves significantly higher EE compared to the reference case at low load. However, this EE gain keeps decreasing with the increase in load.  At the peak load, the gain is insignificant as the probability of having small number of users which allows EE improvement by reducing antennas is very low,  see Fig. \ref{fig:pdf}.

\subsection {Energy  saving potential}
Fig. \ref{fig:Pconsumption} shows the  actual power consumption at different network load for both the adaptive antenna scheme and the fixed antennas, i.e., the reference case. It is observed that power consumption is drastically reduced at low load as the number of active 
\begin{figure}[!tbp]
  \centering
  \begin{minipage}[b]{0.49\textwidth}
    \includegraphics[width=\textwidth]{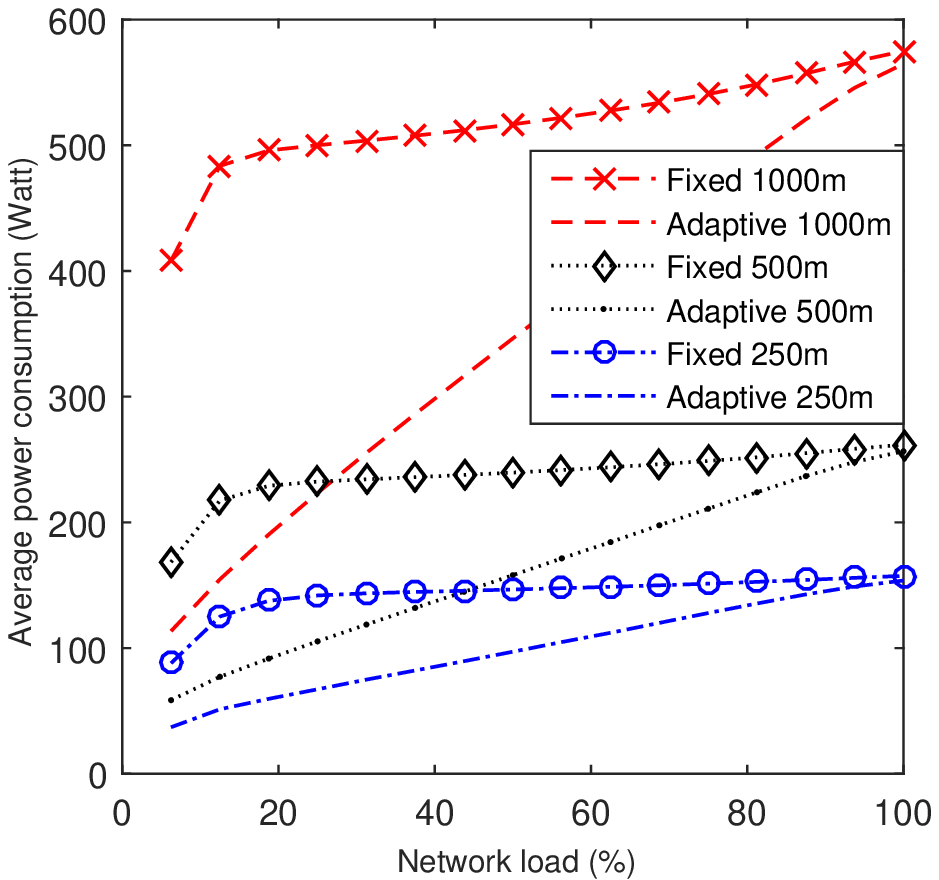}
    \caption{Power consumption vs. network load at different cell sizes for both fixed and adaptive antenna schemes.}
			\label{fig:Pconsumption}
  \end{minipage}
 \hfill
  \begin{minipage}[b]{0.49\textwidth}
	    \includegraphics[width=\textwidth]{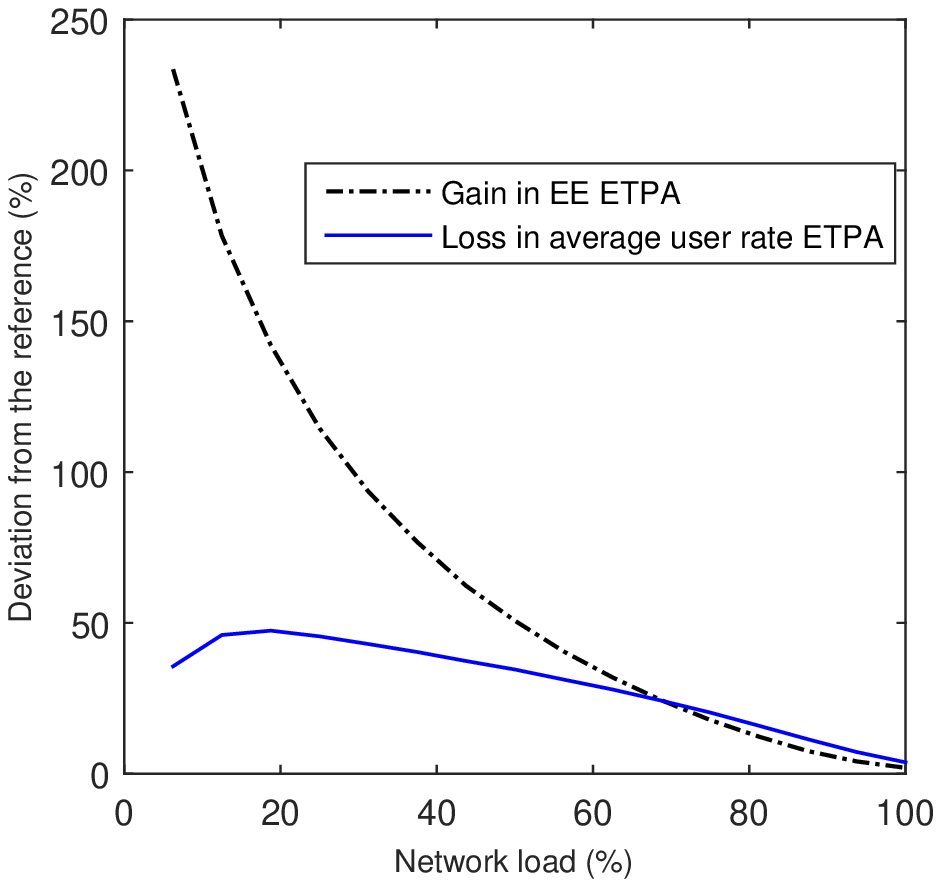}
			    \caption{Tradeoff between EE and user rate at different network loads.}
					\label{fig:TradeOffWS}
  \end{minipage}
\end{figure}
 antennas are allowed to be low following the load under the adaptive scheme. However, the difference becomes less significant with the increase of network load.  Over 24 hour of operation, it is observed that  the  energy saving potential  by adopting adaptive antenna scheme over reference case is \{$40$, $39$, $38$\}\%  for ET-PA and \{$40$, $40$, $37$\}\% for TPA at  \{$1000$, $500$, $250$\} m cell radius. Note that the idle state, i.e., the fraction of time when the BS serves no user is different under these two schemes when serving the same load. Therefore, it has been taken into  account while finding total energy saving potential for fair comparison.

\subsection {EE and user rate tradeoff}
Fig.~\ref{fig:TradeOffWS} shows the tradeoff between the EE and the average user rate at different load for TPA. At very low load the EE has been increased with around 250\% at the cost of around 50\% reduction of the average user data rate. However, with the increase of load in the system, both the gain in EE and loss of user rate get reduced. Over the 24 hour operation, EE has been found to be improved around 24\% at the cost of around 12\% reduction in user rate. Note that the gain in EE and loss in user rate was quite similar for both TPA and ET-PA.

\begin{figure}[!tbp]
  \centering
  \begin{minipage}[b]{0.49\textwidth}
    \includegraphics[width=\textwidth]{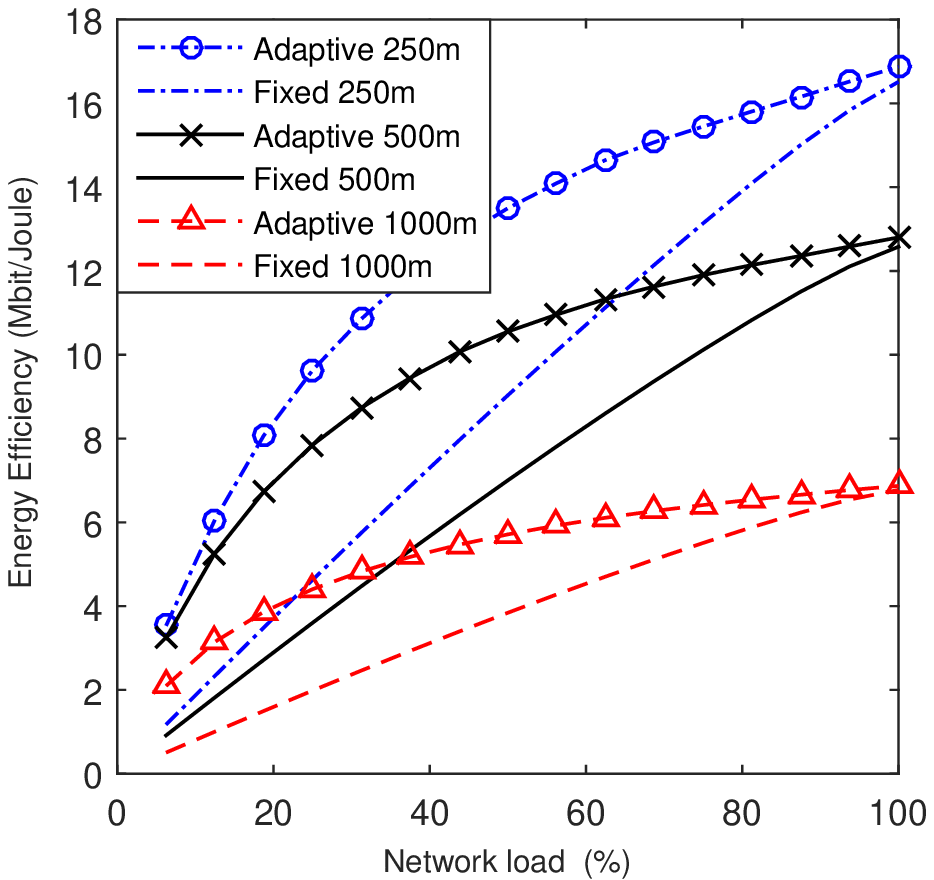}
    \caption{Impact of cell size on EE vs. network Load.}
		\label{fig:dCompareEE}
  \end{minipage}
  \hfill
  \begin{minipage}[b]{0.49\textwidth}
    \includegraphics[width=\textwidth]{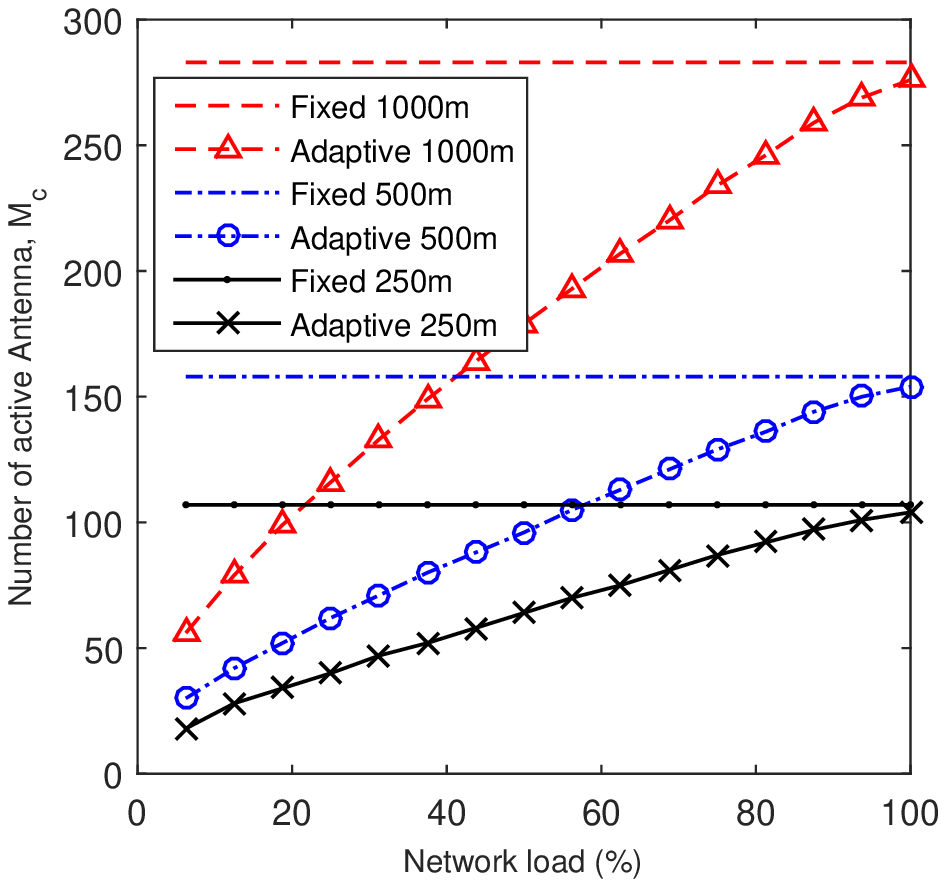}
   \caption{Impact of cell size on number of antennas vs. network Load.}
\label{fig:dCompareAntenna}
  \end{minipage}
\end{figure}
\subsection{Impact of cell size}
In Fig.~\ref{fig:dCompareEE} and  Fig.~\ref{fig:dCompareAntenna}, we show the impact of the cell size on the EE and the corresponding antenna activation policy of the BS. We found that for a cell radius of \{$1000$, $ 500$, $250$\} m, the optimum average transmit power per antenna is  \{$0.21$, $0.101$, $0.04$\} Watt and \{$0.401$, $ 0.183$,  $ 0.065$\} Watt and the corresponding  $M_{\max}$ are  \{$283$, $158$, $107$\} and \{$230$, $134$, $97$\} and $K_{\max}$  are \{$107$, $76$, $56$\} and \{$97$, $68$, $52$\}  for TPA and ET-PA respectively. The corresponding improvement of EE for TPA is \{$22.9$, $22.51$, $27$\}\% and for ET-PA is \{$23.1$, $22.08$, $25.8$\}\%. 
\begin{figure}[!hbp] 
\centering
        \includegraphics[ scale=1]{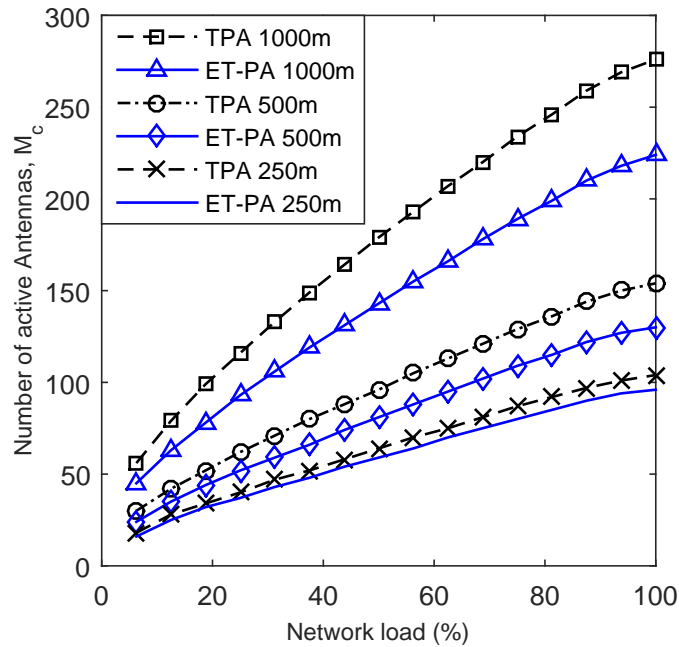}	
\caption{Comparison between TPA and ET-PA in terms of average number of active antennas vs. network load for different cell size.}
\label{fig:TPAvsETPAAntenna}
\end{figure}
 The  EE increases with the decrease in cell size and  $M_{\max}$, $K_{\max}$ and optimum $p$ decrease with the decrease in cell size. 

\begin{figure}[!bp]
  \centering
  \begin{minipage}[b]{0.49\textwidth}
    \includegraphics[width=\textwidth]{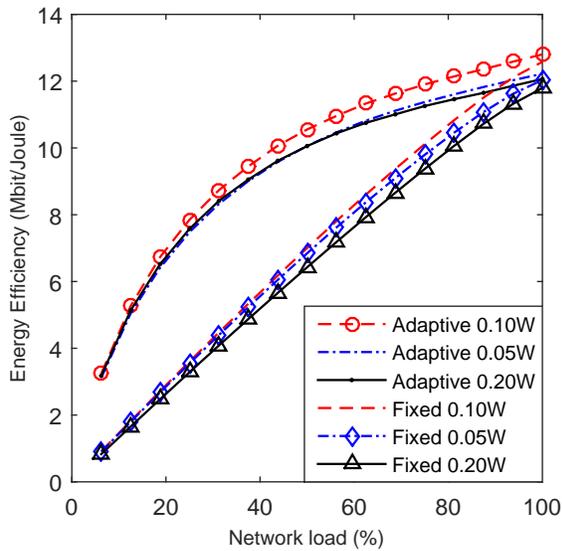}
    \caption{Impact of PA dimensioning on EE vs. network Load.}
		\label{fig:pCompareEET}
  \end{minipage}
  \begin{minipage}[b]{0.49\textwidth}
    \includegraphics[width=\textwidth]{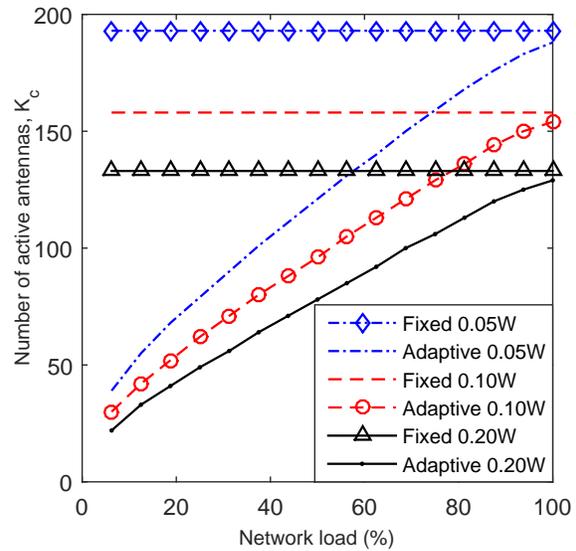}
    \caption{Impact of PA dimensioning on number of active antennas vs. network Load.}
		\label{fig:pCompareEAT}
  \end{minipage}
\end{figure}

\subsection {Impact of  PA efficiency}
It is evident from Fig. \ref{fig:EEtPAeTPA} and Fig. \ref{fig:TPAvsETPAAntenna} that the efficiency of the PA has significant impact on the energy efficient design of the network.   In case of TPA, the efficiency falls drastically when the operating point moves away from  $P_{\max,\text{PA}}$, whereas ET-PA achieves higher efficiency over a wider operating region. Therefore, each BS uses more antennas with TPA compared to ET-PA with less power per antenna as a means to push more number of antennas to operate near the compression region to reap the benefit of higher PA efficiency. Also note that the TPA requires comparatively higher number of  users to be served simultaneously in order to achieve maximum EE. Nevertheless, ET-PA yields higher EE for both the reference case and our scheme compared to TPA for any given cell radius.

\subsection{Impact of PA dimensioning}
 Fig.~\ref{fig:pCompareEET} and Fig.~\ref{fig:pCompareEAT} show the impact of PA dimensioning  on the EE and the corresponding change in required number of antennas for TPA.
The optimum average transmit power of the ET-PA is $0.1$ Watt when the cell radius is $500$m.  For the transmit power \{$0.05$, $0.10$, $0.20$\} Watt, the corresponding overall daily savings are \{$21$, $23$, $25$\}\%. Using PAs with power other than the optimum one reduces the overall EE for both fixed and adaptive antenna cases. However, the energy saving percentages by our scheme increases if the PA uses higher power than the optimum and vice versa. Similar trend  was observed for ET-PA.

 \section{Conclusion}
\label{sec:Conclusion}
In this work we investigate how to dynamically adapt  massive MIMO systems  to the user loads for higher EE. The temporal variation of load has been captured by modeling the load at each BS equipped with massive MIMO functionality  with an $M/G/m/m$ queue and mapping the user distribution to a DLP. We developed a game theory based distributive algorithm that yields significant gains in EE at the cost of reducing the average user data rate at low user load. However, the  rate degradation while increasing the EE  comes from the fact that we consider a very tight reference case. In our reference case, the system considers the complete shutdown of all the antennas when the BS is not serving any user. This reduces the interference significantly resulting in high data rates for some users which in turn allows the BS to reduce its activity time further.  Over the 24 hour operation, EE has been found to be improved around 24\% along with around 40\% energy saving potential while the corresponding reduction in user rates is found to be around 12\%. For transparency, the algorithm was developed for a simple rate formula based on perfect CSI, but the same methodology can be applied to other rate formulas as well.

\section*{Acknowledgment}
This work was  supported by EIT ICT Labs funded  5GrEEn, EXAM;   Academy of Finland funded  FUN5G, and also by ELLIIT and the Swedish Foundation for Strategic Research project.

\begin{appendices}
\section{}
\label{SupermodularProof}
In order to prove that the \emph{EE maximization game} is supermodular, it is sufficient to show that the three requirements stated in Definition $3$ are fulfilled. It is easy to see that the first two conditions hold and it remains to prove that the utility function ${E}_{c}$ has increasing differences. Since the optimal number of active antennas at different user states for a cell are independent of each other, it suffices to show that the function in \eqref{eq:Ebfinal} is supermodular.

\begin{lemma}
 Power term related to user rate, i.e.,  $\mathcal{A} K_c R_c$ in \eqref{eq:TPbb}, in the response dynamics can be ignored when maximizing \eqref{eq:Ebfinal}. 
\end{lemma} 
\begin{proof}
From \eqref{eq:P11}, \eqref{eq:TPbb} and \eqref{eq:Ebfinal}, for a given user $K_c$ our  objective is to maximize the function 
\begin{equation}  
\label{eq:lemma1} 
 \frac{K_c R_c}{\mathcal{A}K_c R_c + C_0^{'} + C_1 M_c}
\end{equation}
Note that maximizing \eqref{eq:lemma1} is equivalent to minimize $ {1+ \frac{ C_0{'} + C_1 M_c}{K_c R_c}}$, which is equivalent to minimize $ \frac{ C_0^{'} + C_1 M_c}{K_c R_c}$.
Therefore, while maximizing \eqref{eq:lemma1}, the term $\mathcal{A}K_c R_c$ can be ignored.
\end{proof}

 Let us denote the utility function of user state $n$ for any cell $c$ by the function 
\begin{equation}
f(x,y) = \frac{R(x,I(y))}{P(x)} 
\end{equation}

where $x=M_c(n)$ is the number of antennas that BS $c$ uses to serve $n$ users,  $I(y)>0$ is the received interference power from any  of the  interfering BSs  operating  with $y$ active antennas plus the noise power and $P(x)$ gives the power consumption by BS $c$ when serving $n$ users with $x$ antennas. If $f(x,y)$ has increasing differences then $\log{f(x,y)}$ will also have increasing differences. Let us define 
\begin{eqnarray}
F(x,y) = \log{f(x,y)}= \Omega(x,y)-\Phi(x)  
\end{eqnarray}
where $\Omega{(x,y)}=\log{R(x, I(y))}$ and $\Phi(x) = \log{P(x)}$.

According to  Definition $2$, it suffices to show that  $F_{xy} \geq 0 $ where $F_{xy}$ is the first cross derivative of $F(x,y)$ with respect to $x$ and $y$. We note that $ F_{xy}=  \Omega_{xy}$. Hence $F(x,y)$ is supermodular  if $\Omega(x,y)$ is supermodular. Note that supermodularity of $F(x,y)$ is independent of the PA characteristics $\Phi(x)$. 

 Considering the array gain when the BS  serves $n$ users simultaneously with $x= M_c$ antennas, $\Omega(x,y)$ can be written as $\log\log\Gamma(x,y)$  where $\Gamma(x,y)=1+\frac{x (x-n)}{I(y)}$ according to  \eqref{eq:Rc} and \eqref{eq:Ebfinal}. As $I(y) > 0$, with the help of the first derivatives,  $\Gamma_{x} $, $\Gamma_{y}$ and the first cross derivative, $\Gamma_{xy} $, we can express $\Omega_{xy}$ as 
\begin{eqnarray}
\label{crderivative1}
\Omega_{xy} &=& - \frac{1}{\log\Gamma(x,y)} \frac{1}{\Gamma(x,y)} \cdot \nonumber \\ 
 && \!\left(\frac{1}{\Gamma(x,y)}\!\left(\frac{1}{\log\Gamma(x,y)}\!+\! 1\!\right)\!\Gamma_{x}\Gamma_{y}\!-\!\Gamma_{xy}\!\right).
\end{eqnarray}
Based on \eqref{crderivative1}, $\Omega_{xy}$ will be nonnegative if
\[ 
 \frac{1}{\Gamma(x,y)} \left(\frac{1} {\log\Gamma(x,y)} + 1
 \right)\Gamma_{x} \Gamma_{y} - \Gamma_{xy} \leq 0.
\]
After computing the cross derivatives, the above inequality can be written as 
\begin{equation}
\label{inequal1}
\frac{1}{\Gamma(x,y)}\left(\frac{1} {\log\Gamma(x,y)}+1\right)\frac{x(x-n)}{I(y)} \geq 1.
\end{equation} 
After replacing $\Gamma(x,y)= 1+\frac{x(x-n)}{I(y)}$ in  inequality~\eqref{inequal1}, we get 
\begin{equation} 
\label{inequal2}
\left(1+ \frac{x(x-n)}{I(y)}\right) \leq \left(\frac{1}{\log\left(1+ \frac{x(x-n)}{I(y)}\right)}+1 \right)\frac{x(x-n)}{I(y)}.
\end{equation}
After canceling out common terms inequality~\eqref{inequal2} becomes
$\frac{x(x-n)}{I(y)} \geq \log\left(1+\frac{x(x-n)}{I(y)}\right)$, which holds.
We can deduce that $\Omega(x,y) $  is supermodular since $z \geq \log(1+z)$ for $z= x(x-n)/I(y)$. Hence EE maximization problem is supermodular.

\section{}
\label{RateFormulaProof}
In order to derive the rate formula \eqref{eq:Rc}, let us denote the desired signal received by the user $k$ of cell $c$ by
\begin{equation}
\label{RecievedSignal}
y_{ck} =\mathbf{h}_{cck}^H \sum_{i=1}^{K_c}{\mathbf{w}_{ci} s_{ci}}+  \sum_{d \neq c} \mathbf{h}_{dck}^H \sum_{i=1}^{K_d}{\mathbf{w}_{di}s_{di}.}
\end{equation}
Assuming perfect CSI at the BS and users and using ZF precoding, (i.e., $\mathbf{h}_{cck}^H \mathbf{w}_{ci}=0$ for $k \neq i)$ and considering equal power allocation, i.e., $s_{ci} \sim \mathcal{CN}(0,  \frac{ p M_c}{K_c})$, the ergodic rate of user $k$ in cell $c$

\begin{eqnarray}
\label{eq:RateFormula1}
R_{ck} &=& \beta \mathbb{E}\left[ \log_2 \left( 1+ \frac{p \frac{M_c}{K_c} |\mathbf{h}_{cck}^H \mathbf{w}_{ck}|^2}{\sigma^2+\sum_{d \neq c}p\frac{M_d}{K_d} \sum_{i=1}^{k_d} |\mathbf{h}_{dck}^H \mathbf{w}_{dk} |^2} \right)\right]  \\
&\geq & \beta \log_2 \left( 1+ \frac{p \frac{M_c}{K_c} (\mathbb{E} [|\mathbf{h}_{cck}^H \mathbf{w}_{ck}|^{-2}] )^{-1}}{\sigma^2+\sum_{d \neq c}p \frac{M_d}{K_d} \sum_{i=1}^{k_d} \mathbb{E} \left[|\mathbf{h}_{dck}^H \mathbf{w}_{dk} |^2\right]} \right) \label{eq:Jensen}
\end{eqnarray} 
where Jensen's inequality is used to derive \eqref{eq:Jensen} from \eqref{eq:RateFormula1} and $\beta= B\left(1-\frac{\alpha K_{\max}}{T_c}\right)$. \\

Assuming $\mathbf{h}_{dci} \sim \mathcal{CN}(0, g_{dci} \mathbf{I}_{M_d})$, $\mathbb{E} \left[ |\mathbf{h}_{dck} \mathbf{w}_{dk}|^2\right] =g_{dck} \cdot 1 = g_{dck}$ as $\mathbf{h}_{dck}$ and $\mathbf{w}_{dk}$ are independent for $d \neq c$ and   $\mathbb{E}[||\mathbf{w}_{dk}||^2]=1$. Also 
$\mathbb{E}[|\mathbf{h}_{cck} \mathbf{w}_{ck}|^{-2}] = \frac{1}{(M_c-K_c) g_{cck}}$ leads to
\begin{equation}
\label{RateFormula2}
R_{ck}= \beta \log_2 \left( 1+ \frac{p \frac{M_c}{K_c} (M_c-K_c) g_{cck}}{\sigma^2+\sum_{d \neq c}p\frac{M_d}{K_d} \sum_{i=1}^{k_d} g_{dck}} \right).
\end{equation}
 \\ Finally, we compute a lower bound on the average ergodic rate in a cell by using Jensen's inequality as
\begin{eqnarray}
\label{RateFormula3}
\mathbb{E}[R_{ck}] & \geq & \beta \log_2 \left( 1+ \frac{p \frac{M_c}{K_c} (M_c-K_c)}{\sigma^2 \mathbb{E}{[\frac{1}{g_{cck}}]}+\sum_{d \neq c}p\frac{M_d}{K_d} K_d\mathbb{E} [\frac{g_{dck}}{g_{cck}}] }\right)\\
&=& B\left(1-\frac{\alpha K_{\max}}{T_c}\right) \log_2 \left(1+ \frac{p\frac{M_c}{K_c}(M_c-K_c)}{\sigma^2 G_{cc}+\sum_{d \neq c}p M_d G_{cd}} \right)
\end{eqnarray}
where we define $G_{cc}$= $\mathbb{E}$ ${[\frac{1}{g_{cck}}]}$  and $G_{cd}$=$\mathbb{E}$${[\frac{g_{dck}}{g_{cck}}]}$.


\end{appendices}

\end{document}